\documentclass[11pt,a4paper]{article}
\usepackage[margin=2cm]{geometry}
\usepackage[utf8]{inputenc}
\usepackage{amsmath}
\usepackage{amssymb}
\usepackage{amsthm}
\usepackage{dsfont}
\usepackage{color}
\usepackage{amsmath}
\usepackage{amsfonts}
\usepackage{amssymb}
\usepackage{bbm}
\usepackage{lineno}
\usepackage{relsize}
\usepackage{pst-node}
\usepackage[procnumbered, linesnumbered,algoruled]{algorithm2e}
\usepackage[compact]{titlesec}
\usepackage{tikz}
\usetikzlibrary{positioning}
\usetikzlibrary {shapes,matrix}
\usetikzlibrary{arrows}
\tikzset{
	semithick,
	node distance = 2cm,
	dot/.style={circle,fill,inner sep=2pt}
}
\tikzset{
	side by side/.style 2 args={
		line width=2pt,
		#1,
		postaction={
			clip,postaction={draw,#2}
		}
	}
}
\tikzstyle{every state}=[draw = black,thick,fill = white,minimum size = 4mm]
\tikzstyle{selected edge} = [draw,line width=2pt,-,red!50]
\tikzset{
	edge/.style={->,> = latex'}
}

\newcommand{\comment}[1]{}

\usepackage[procnumbered, linesnumbered,algoruled]{algorithm2e}
\newcommand{\cm}{{\mathcal{M}}}
\newcommand{\cA}{{\mathcal{A}}}

\newcommand{\cI}{{\mathcal{I}}}
\newcommand{\cF}{{\mathcal{F}}}

\newcommand{\OPT}{\textnormal{OPT}}

\newcommand{\eps}{{\varepsilon}}

\newcommand{\floor}[1]{\left\lfloor #1 \right\rfloor}

\DeclareMathOperator*{\argmin}{arg\,min}

\begin{document}
\newtheorem{thm}{Theorem}[section]
\newtheorem{prop}[thm]{Proposition}
\newtheorem{assm}[thm]{Assumption}
\newtheorem{lem}[thm]{Lemma}
\newtheorem{obs}[thm]{Observation}
\newtheorem{cor}[thm]{Corollary}
 \newtheorem{lemma}[thm]{Lemma}
  \newtheorem{definition}[thm]{Definition}
 \newtheorem{theorem}[thm]{Theorem}
 \newtheorem{proposition}[thm]{Proposition}
 \newtheorem{claim}[thm]{Claim}
\newtheorem{defn}[thm]{Definition}
\newcommand{\ariel}[1]{{\color{red} (Ariel :#1)}}
\def \II   {{\mathcal I}}
\newcommand{\one}{\mathbbm{1}}
	\def\claimproof{\proof}
\def\endclaimproof{\hfill$\square$\\}
\renewcommand\qedsymbol{$\blacksquare$}

\title{
An FPTAS for Budgeted Laminar Matroid Independent Set
}
\author{Ilan Doron-Arad\thanks{Computer Science Department, 
	Technion, Haifa, Israel.\texttt{idoron-arad@cs.technion.ac.il}}
\and	
	Ariel Kulik\thanks{CISPA Helmholtz Center for Information Security, Germany. \texttt{ariel.kulik@cispa.de}} 
\and 
Hadas Shachnai\thanks{Computer Science Department, 
Technion, Haifa, Israel. \texttt{hadas@cs.technion.ac.il}}
}
\maketitle

\begin{abstract}
	
	We study the {\em budgeted laminar matroid independent set} problem. The input is a ground set, where each element has a cost and a non-negative profit, along with a laminar matroid over the elements and a budget. The goal is to select a maximum profit independent set of the matroid whose total cost is bounded by the budget. 
	Several well known special cases, where we have, e.g., no matroid constraint (the classic knapsack problem) or a uniform matroid constraint (knapsack with a cardinality constraint), admit a {\em fully polynomial-time approximation scheme (FPTAS)}. In contrast, the {\em budgeted matroid independent set (BMI)} problem with a general matroid has an {\em efficient polynomial-time approximation scheme (EPTAS)} but does not admit an FPTAS. 
	This implies an EPTAS for our problem, which is the best known result prior to this work.

	 We present an FPTAS for budgeted laminar matroid independent set, improving the previous EPTAS for this matroid family and generalizing the FPTAS known for knapsack with a cardinality constraint and multiple-choice knapsack. 
	 Our scheme is based on a simple {\em dynamic program}  
	  which utilizes the tree-like structure of laminar matroids.


\end{abstract}

\section{Introduction}
\label{sec:intro}

\comment{@book{martello1990knapsack,
		title={Knapsack problems: algorithms and computer implementations},
		author={Martello, Silvano and Toth, Paolo},
		year={1990},
		publisher={John Wiley \& Sons, Inc.}
	}
	
	@book{kellerer2004multidimensional,
		title={Multidimensional knapsack problems},
		author={Kellerer, Hans and Pferschy, Ulrich and Pisinger, David and Kellerer, Hans and Pferschy, Ulrich and Pisinger, David},
		year={2004},
		publisher={Springer}
	}
	
	@article{gilmore1966theory,
		title={The theory and computation of knapsack functions},
		author={Gilmore, PC and Gomory, Ralph E},
		journal={Operations Research},
		volume={14},
		number={6},
		pages={1045--1074},
		year={1966},
		publisher={INFORMS}
	}
	
	@article{cacchiani2022knapsack,
		title={Knapsack problems-An overview of recent advances. Part II: Multiple, multidimensional, and quadratic knapsack problems},
		author={Cacchiani, Valentina and Iori, Manuel and Locatelli, Alberto and Martello, Silvano},
		journal={Computers \& Operations Research},
		pages={105693},
		year={2022},
		publisher={Elsevier}
}}

Knapsack is one of the most fundamental problems in combinatorial optimization, which
has been continuously studied in the past half century.
 \cite{gilmore1966theory,martello1990knapsack,kellerer2004multidimensional,cacchiani2022knapsack}. Considerable attention was given to a generalization of knapsack including an additional {\em matroid constraint} \cite{DBLP:journals/cacm/Knuth74,caprara2000approximation,mastrolilli2006hybrid,li2022faster,sinha1979multiple,bansal2004improved,kellerer2004,doron2022eptas}. 
In this work we consider the knapsack problem with a {\em laminar matroid} constraint. 

\comment{These problems have numerous applications arising in, e.g., assortment planning \cite{desir2016assortment}, crowdsourcing \cite{wu2015hear}, multiprocessor task scheduling \cite{caprara2000approximation},  capital budgeting \cite{nauss19780}, and fault tolerance optimization \cite{sinha1979multiple}. 

Another interest in knapsack problems with additional constraints is as {\em separation oracles}, a crucial component in the ellipsoid method used in solving linear programs with a large number of constraints \cite{karmarkar1982efficient,fleischer2011tight,epstein2012bin}. Surprisingly, for slightly more general families of matroid constraints than discussed above, such as partition matroid constraints, we could not find any published work in the context of knapsack; we encountered this while writing a different paper, which required such a separation oracle \cite{doron2022bin}. 
}

A {\em matroid} is a set system $(S, \cI)$, where $S$ is a finite set and $\cI \subseteq 2^S$ such that (i) $\emptyset \in \cI$, (ii) for all $A \in \cI$ and $B \subseteq A$ it holds that $B \in \cI$, and (iii) for all $A,B \in \cI$ where $|A| > |B|$ there is $e \in A \setminus B$ such that $B \cup \{e\} \in \cI$. We focus on the family of {\em laminar matroids}, defined below. 
 \begin{definition}
	Given a finite set $S$, $\cF \subseteq 2^S \setminus \{\emptyset\}$ is a {\em laminar family} on $S$ if for any $X,Y \in \cF$ one of the following holds:  $X \cap Y  = \emptyset$, or $X \cap Y  = X$, or $X \cap Y  = Y$. 
\end{definition} 

\begin{definition}
Let $\cF$ be a laminar family on a finite set $S$;
also, let $k: \cF \rightarrow \mathbb{N}_{>0}$ and $\cI_{\cF,k} = \{A \subseteq S~:~|A \cap X| \leq k(X)~\forall X \in \cF\}$. Then $(S,\cI_{\cF,k})$ is a {\em laminar matroid}. 
\end{definition} 
\noindent The {\em independent sets} $\cI_{\cF,k}$ in the laminar matroid $(S,\cI_{\cF,k})$ are all subsets of elements $A$ such that for any set $X \in \cF$ in the laminar family, $|A \cap X|$ does not violate the cardinality constraint $k(X)$. 
It is well known (see, e.g., \cite{IW11,G01}) that laminar matroids are indeed matroids.

In this paper, we study the {\em budgeted laminar matroid independent set (BLM)} problem. The input is a tuple $I = (S,\cF,k,c,p,B)$, where $S$ is a finite set, $\cF$ is a laminar family on $S$ such that $S \in \cF$, $k:\cF \rightarrow \mathbb{N}_{>0}$ gives cardinality bounds to $\cF$, $c:S \rightarrow \mathbb{N}$ is a cost function, $p:S \rightarrow \mathbb{N}$ is a profit function, and $B \in \mathbb{N}$ is a budget.  A {\em solution} of $I$ is an independent set $T \in \cI_{\cF,k}$ of the laminar matroid $(S, \cI_{\cF,k})$ such that $c(T) = \sum_{e \in T} c(e) \leq B$. The goal is to find a solution $T$ of $I$ such that $p(T) =  \sum_{e \in T} p(e) $ is maximized. 

There are a few notable special cases of BLM; we use our notation to formally define them. Consider an instance $I$ of BLM. If $\cF \setminus \{S\}$ is a partition of $S$ and $k(S) \geq |S|$ then $(S, \cI_{\cF,k})$ is a {\em partition matroid}, and we say that $I$ is a {\em knapsack with partition matroid} instance.\footnote{
	Adding the constraint $k(S) \geq |S|$ is purely technical to instantiate a partition matroid using BLM notation.} Moreover, the special case of knapsack with partition matroid where $k(X) = 1~\forall X \in (\cF \setminus \{S\})$ is the {\em multiple-choice knapsack} \cite{sinha1979multiple}. Alternatively, if $\cF = \{S\}$ then $(S,\cI_{\cF,k})$ is a {\em uniform matroid}, and $I$ is a {\em cardinality constrained knapsack} instance \cite{DBLP:journals/cacm/Knuth74}. 

A natural application of BLM arises in the context of cloud computing, where limited network bandwidth limit plays a vital role (see, e.g.,  \cite{zhu2012towards,popa2013elasticswitch,popa2012faircloud}). Consider a network $G = (V,E)$, that is a directed tree with a root $M$ associated with a cloud computer. 
Each leaf in $G$ represents a {\em client}. 
Each client sends a job to $M$, which can process the job and broadcast the results through the network back to the client.  
The jobs have processing times and values. Also, $M$ has a bound on the total processing time of admitted jobs. Moreover, each node 
$v \in V$ in the network has a limited bandwidth, $k(v)$, so that $v$ can broadcast the results of at most $k(v)$ jobs to its descendant clients. The goal is to maximize the total value of complete jobs subject to the processing time and bandwidth bounds. 

We can cast this problem as a BLM instance by taking the set of elements $S$ to be the se of jobs, where the profit and cost of each element are the value and processing time of the corresponding job. Now, we define
 a laminar family $\cF$ on the set of elements, where for each vertex $v \in V$ there is a set $X_v \in \cF$ containing all jobss that belong to clients (i.e., leaves) in the subtree rooted by $v$, where the cardinality bound of $X_v$ is the bandwidth limit $k(X_v) = k(v)$. For other applications of laminar matroids, see e.g.,~\cite{G01,GMT15,IW11}.

We focus in this paper on approximation schemes for BLM. Let $\OPT(I)$ be the value of an optimal solution for an instance $I$ of a maximization problem $\Pi$. For $\alpha \in (0,1]$, we say that $\cA$ is an $\alpha$-approximation algorithm for $\Pi$ if, for any instance $I$ of $\Pi$,
$\cA$ outputs a solution of value at least $\alpha \cdot \OPT(I)$. A {\em polynomial-time approximation scheme} (PTAS)
for $\Pi$ is a family of algorithms $(A_{\eps})_{\eps>0}$ such that, for any $\eps>0$, $A_{\eps}$ is a polynomial-time $(1 - \eps)$-approximation algorithm for $\Pi$.
An {\em efficient PTAS} (EPTAS) is a PTAS $(A_{\eps})_{\eps>0}$ with running time of the form $f\left(\frac{1}{\eps}\right) \cdot |I|^{O(1)}$, where $f$ is an arbitrary computable function.
 The running time of an EPTAS might be impractically high; this motivates the study of the following important subclass of EPTAS: $(A_{\eps})_{\eps>0}$ is a {\em fully PTAS} (FPTAS) if the running time of $A_{\eps}$ is of the form ${\left(\frac{|I|}{\eps}\right)}^{O(1)}$. 
 
 \begin{figure}\hspace*{0.8cm}
 	\label{fig:diagram}
 	\begin{tikzpicture}[
 		terminal/.style={
 			rectangle,
 			minimum width=2cm,
 			minimum height=1cm,
 			very thick,
 			draw=black,
 			font=\itshape,
 		},
 		]
 		\matrix[row sep=0.7cm,column sep=-5cm] {%
 			&; \node [terminal](KP) {$0/1$-knapsack {\bf FPTAS}};  &\\
 			
 			\node [terminal](MCK) {Multiple-choice knapsack {\bf FPTAS}}; & & \node [terminal](kKP) {Cardinality constrained knapsack {\bf FPTAS}}; \\
 			&\node [terminal](PKP) {Knapsack with partition matroid {\bf FPTAS}}; &\\
 			&\node [terminal,draw=red](BLM) {Budgeted laminar matroid independent set (BLM) {\bf FPTAS} \textnormal {(this paper)}}; \\
 			& \node [terminal](BMI) {Budgeted matroid independent set {\bf EPTAS},  no {\bf FPTAS} unless P=NP}; \\
 		};
 		\draw (KP) edge [->,>=stealth,shorten <=2pt, shorten >=2pt, thick] (kKP);
 		\draw (KP) edge [->,>=stealth,shorten <=2pt, shorten >=2pt, thick] (MCK);
 		\draw (MCK) edge [->,>=stealth,shorten <=2pt, shorten >=2pt, thick] (PKP);
 		\draw (kKP) edge [->,>=stealth,shorten <=2pt, shorten >=2pt, thick] (PKP);
 		\draw (PKP) edge [->,>=stealth,shorten <=2pt, shorten >=2pt, thick] (BLM);
 		\draw (BLM) edge [->,>=stealth,shorten <=2pt, shorten >=2pt, thick] (BMI);
 	\end{tikzpicture}
 	\caption{Complexity of knapsack generalizations with various families of matroid constraints. }
 \end{figure}
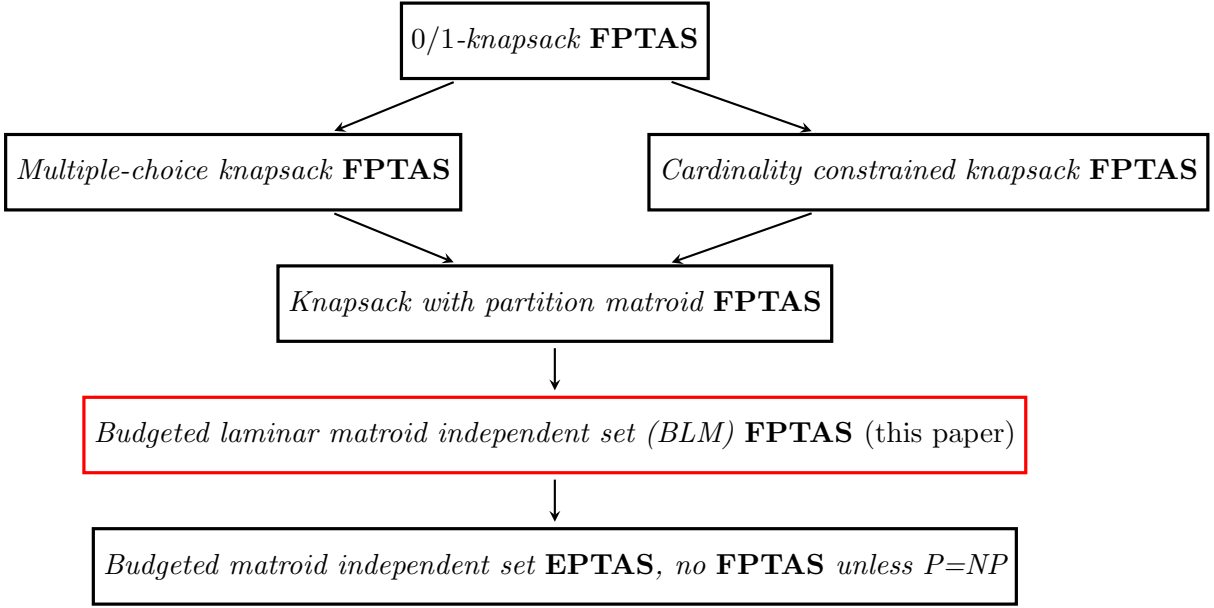
 
We note that budgeted independent set with a {\em general} matroid constraint admits an EPTAS; however, the existence of an FPTAS was ruled out \cite{MM2023}. In contrast, well known special cases of BLM such as cardinality constrained knapsack \cite{DBLP:journals/cacm/Knuth74} and multiple-choice knapsack~\cite{sinha1979multiple} admit FPTAS. The question whether 
budgeted independent set admits an FPTAS on other classes of matroids (e.g., laminar, graphic, or linear matroid) remained open. 

In this paper we resolve this open question for laminar matroids.
Our main result is an FPTAS for BLM, improving upon the existing EPTAS for this family, and generalizing the FPTAS for the special cases of cardinality constrained knapsack and multiple-choice knapsack. Specifically,  

\begin{theorem}
	\label{thm:main}
There is an algorithm \textnormal{\textsf{FPTAS}} that given a \textnormal{BLM} instance $I$ and $\eps>0$ finds in time $O\left( |I|^5 \cdot \eps^{-2} \right)$ a solution $T$ for $I$ of profit $p(T) \geq (1-\eps) \cdot \OPT(I)$.    
\end{theorem}

We give in Figure~\ref{fig:diagram} a complexity overview of the related problems. To derive an FPTAS, we first find a pseudo-polynomial time algorithm for BLM. Our technique is based on {\em dynamic programming (DP)}, which exploits the tree-like structure of laminar matroids. More concretely, let 
  $I =  (S,\cF,k,c,p,B)$ be a BLM instance and $X \in \cF$ a {\em maximal set} in the laminar family, not contained in any other set except for $S$. We first 
  construct (recursively) DP tables 
  for elements contained in $X$ and separately for elements contained in $S \setminus X$. Then, we combine the two DP tables into a DP table for the original instance $I$.
  We rely on the key property that combining any two independent sets $T_1 \subseteq X$ and independent set $T_2 \subseteq S \setminus X$ yields an independent set $T_1 \cup T_2$ as long as the overall cardinality (i.e., $|T_1 \cup T_2|$) satisfies the cardinality bound of $S$. Finally, our FPTAS is obtained by standard rounding of the profits. We remark that we did not attempt to optimize the running time; instead, our goal is to obtain a simple FPTAS for the problem.   
\subsection{Related Work}

BLM is an immediate generalization of the classic $0/1$-knapsack problem. While the knapsack problem is known to be NP-hard, it admits an FPTAS. Other notable special cases of BLM that admit an FPTAS include cardinality constrained knapsack \cite{DBLP:journals/cacm/Knuth74,caprara2000approximation,mastrolilli2006hybrid,li2022faster} and multiple choice knapsack \cite{sinha1979multiple, bansal2004improved,kellerer2004}. 


Chakaravarthy et al. \cite{chakaravarthy2013knapsack} considered the {\em knapsack cover with a matroid constraint (KCM)} problem,  which is dual 
to budgeted matroid independent set. In this variant, we are given a matroid $\cm = (S,\cI)$, a cost function $c:S \rightarrow \mathbb{R}_{\geq 0}$, a size function $s:S \rightarrow \mathbb{R}_{\geq 0}$, and a demand $D \in \mathbb{R}_{\geq 0}$. The goal is to find $A \in \cI$ such that $s(A) = \sum_{e \in A} s(e) \geq D$ 
such that $c(A) = \sum_{e \in A} c(e)$ is minimized.
They obtain a PTAS for a general matroid $\cm$, and an FPTAS when $\cm$ is a partition matroid. The {\em constrained minimum spanning tree} problem is a special case of KCM where $(S,\cI)$ is a {\em graphic matroid}.\footnote{$(S,\cI)$ is a graphic matroid if there is a graph $G = (V,S)$ where $\cI$ contains all the subsets $T \subseteq S$ satisfying $G' = (V,T)$ is an acyclic graph.}  
The constrained minimum spanning tree problem admits an EPTAS \cite{hassin2004efficient} and an FPTAS which violates the budget constraint by a factor of $(1+\eps)$ \cite{hong2004fully}.

As BLM is a special case of the budgeted matroid independent set problem, a PTAS for the problem follows from the works of \cite{GZ10,chekuri2011multi,BBGS11}. An EPTAS for BMI
was recently presented in~\cite{doron2022eptas} and generalized for budgeted matroid intersection in \cite{doron2023eptas}. As shown in \cite{MM2023}, this is the best possible. 

\subsection{Organization of the Paper} 

Section~\ref{sec:prel} gives preliminary results and notations. In Section~\ref{sec:alg} we give a pseudo-polynomial time algorithm for BLM. 
Section~\ref{sec:FPTAS} presents our FPTAS. We conclude in Section~\ref{sec:discussion} with a discussion. 

\section{Preliminaries}
\label{sec:prel}

Given a set $S$, a function $f:S \rightarrow  \mathbb{N}$, and $R \subseteq S$, let $f(R) = \sum_{e \in R} f(e)$.  Also, we sometimes use $f$ to denote a restriction of $f$ to a subset of the domain $S' \subseteq S$. Let $I = (S,\cF,k,c,p,B)$ be a BLM instance; we use $\cI(I) = \cI_{\cF,k}$ to denote the independent sets of the instance. 
For some $G \subseteq S$ define $\cF_{\subseteq G} = \{Y \in \cF~|~ Y \subseteq G\}$ as all sets in $\cF$ that are contained in $G$. We define below operations on $I$, 
generating modified BLM instances. We say that a set $X \in \cF \setminus \{S\}$ is a {\em maximal set} if it is not contained in any other set in $\cF \setminus \{S\}$.  For a maximal set $X \in \cF \setminus \{S\}$, define the BLM instances $I \cap X = (X, \cF_{\subseteq X},k,c,p,B)$ 
and $I \setminus X = \big(S \setminus X, \cF_{\subseteq (S \setminus X)} \cup \{S \setminus X\},\bar{k},c,p,B\big)$ such that for all $Y \in  \cF_{\subseteq (S \setminus X)} \cup \{S \setminus X\}$ 
 \begin{equation*}
	\bar{k}(Y) =
	\begin{cases}
	k(Y) ~~ &  Y \in \cF \\
		k(S), ~~& \textnormal{otherwise}.
	\end{cases}
\end{equation*} Note that the only set $Y$ that potentially does not belong to $\cF$ is $Y = S \setminus X$. Observe that $I \cap X$ and $I \setminus X$ can be viewed as restrictions of $I$ to $X$ and $S \setminus X$, respectively. The next observations follow from the above definitions and the properties of laminar families. 
\begin{obs}
	\label{obs:cap}
	For a \textnormal{BLM} instance $I = (S,\cF,k,c,p,B)$ and a maximal set $X \in (\cF \setminus \{S\})$ it holds that $I \cap X$ and $I \setminus X$ are \textnormal{BLM} instances. 
\end{obs}

We use $\uplus$ to denote disjoint union. 
\begin{obs}
	\label{obs:cap2}
	For a \textnormal{BLM} instance $I = (S,\cF,k,c,p,B)$ and a maximal set $X \in (\cF \setminus \{S\})$ the following holds. 
	\begin{enumerate}
		\item For any $ Q_1 \in \cI(I \cap X)$ and  $Q_2 \in \cI(I \setminus X)$ such that $|Q_1 \uplus Q_2| \leq k(S)$ it holds that $Q_1 \uplus Q_2 \in \cI(I)$. 
		\item For any $Q \in \cI(I)$ it holds that $Q \cap X \in \cI(I \cap X)$ and $Q \setminus X \in  \cI(I \setminus X)$. 
	\end{enumerate} 
	
\end{obs}

\section{A Pseudo-polynomial Time Algorithm}
\label{sec:alg}

In this section we give a pseudo-polynomial time algorithm for BLM. Let $I = (S,\cF,k,c,p,B)$ be a BLM instance, and define $P_I = \{0,1,\ldots, |S| \cdot \max_{e \in S} p(e)\}$. Observe that for any solution $T \subseteq S$ it holds that $p(T) \in P_I$. 
Also, there may be $t \in P_I$ such that $t \neq p(T)$ for any $T \subseteq S$. Our algorithm computes the following table. 

\begin{definition}
	\label{def:DP}
	For any \textnormal{BLM} instance $I = (S,\cF,k,c,p,B)$, define the \textnormal{DP} table of $I$ as the function $\textnormal{\textsf{DP}}_I: \{0,....|S|\} \times P_I \to \mathbb{N}\cup \infty$ such that for all $q \in \{0,1,\ldots, |S|\}$ and $t \in P_{I}$, 
	$$\textnormal{\textsf{DP}}_I(q,t) = \min_{Q \in \cI_{\cF,k} \textnormal{ s.t. } |Q| = q,~ p(Q) = t} c(Q).$$ 
\end{definition} For any $q \in \{0,1,\ldots, |S|\}$ and $t \in P_{I}$, the entry $\textsf{DP}_I(q,t)$ gives the minimum cost of an independent set in $\cI_{\cF,k}$ of exactly $q$ elements and profit equal to $t$; if there is no such independent set then $\textsf{DP}_I(q,t) = \infty$. Also, $\textsf{DP}_I(q,t) = \infty$ if $q$ or $t$ do not belong to the domain of $\textsf{DP}_I$ (e.g., $q = |S|+1$). 

We formulate a dynamic program which computes the table $\textsf{DP}_I$. Informally, the table $\textsf{DP}_I$ is constructed by taking a maximal set $X \in \cF \setminus \{S\}$ in the laminar family (recall that a maximal set is not a subset of any other set in  $\cF \setminus \{S\}$). Our algorithm computes the sub-tables $\textsf{DP}_{I \cap X}$ and $\textsf{DP}_{I \setminus X}$ recursively. An important observation is that the instances $I \cap X$ and $I \setminus X$ are disjoint;
thus, the table $\textsf{DP}_I$ can be computed from $DP_{I\cap X}$ and $DP_{I\setminus X}$ using a convolution. 
Specifically, to compute an entry $\textsf{DP}_I(q,t)$ for $q \in  \{0,1,\ldots, |S|\}$ and $t \in P_I$, we find the minimum solution induced by partitioning the values $q$ and $t$ between the complementary sub-instances $I \cap X$ and $I \setminus X$. This is formalized by the next lemma. 

\begin{lemma}
	\label{lem:DP-step}

	Let $I = (S,\cF,k,c,p,B)$ be a \textnormal{BLM} instance, and $X \in \cF \setminus \{S\}$ a maximal set. 
	%
	%
	Then, for all $q \in  \{0,1,\ldots, |S|\}$ and $t \in P_I$,
	 \begin{equation*}
		\textnormal{\textsf{DP}}_I(q,t) =
		\begin{cases}
			\min_{\substack{ q_1,q_2,t_1,t_2 \in \mathbb{N} \textnormal{ s.t. } \\ q_1+q_2 = q, ~t_1+t_2 = t }}   \left\{ \textnormal{\textsf{DP}}_{I \cap X}(q_1,t_1) + \textnormal{\textsf{DP}}_{I \setminus X}(q_2,t_2) \right\} ~~ & \textnormal{if}\ q \leq k(S) \\
			\infty ~~& \textnormal{otherwise}
		\end{cases}
	\end{equation*}
\end{lemma}

\begin{proof}
	Let $q \in  \{0,1,\ldots, |S|\}$ and $t \in P_I$. For simplicity, 
	let $$M = \min_{\substack{ q_1,q_2,t_1,t_2 \in \mathbb{N} \textnormal{ s.t. } \\ q_1+q_2 = q, ~t_1+t_2 = t }}   \left\{ 	\textnormal{\textsf{DP}}_{I \cap X}(q_1,t_1) + \textnormal{\textsf{DP}}_{I \setminus X}(q_2,t_2) \right\}.$$ We first consider the case where $\textnormal{\textsf{DP}}_I(q,t)$ and $M$ differ from $\infty$. We use the following auxiliary claims. 	\begin{claim}
		\label{claim:1}
		If $\textnormal{\textsf{DP}}_I(q,t) \neq \infty$ then $M \leq \textnormal{\textsf{DP}}_I(q,t)$. 
	\end{claim}
	\begin{claimproof}
		As $\textnormal{\textsf{DP}}_I(q,t) \neq \infty$  there is $Q^* \in \cI_{\cF,k}$ such that
		\begin{equation}
			\label{eq:arg}
			Q^* =  \argmin_{Q \in \cI_{\cF,k} \textnormal{ s.t. } |Q| = q, ~p(Q) = t} c(Q).
		\end{equation}  Let $q^*_1 = |Q^* \cap X|, q^*_2 = |Q^* \setminus X|, t^*_1 = p(Q^* \cap X)$, and $t^*_2 = p(Q^* \setminus X)$. By \eqref{eq:arg}, and since $S = X \uplus (S \setminus X)$, it holds that $$|Q^*| =  |Q^* \cap X|+|Q^* \setminus X| = q^*_1+q^*_2 = q$$ and $$p(Q^*) = p(Q^* \cap X)+p(Q^* \setminus X) = t^*_1+t^*_2 = t.$$ Moreover, by Observation~\ref{obs:cap2}, since $Q^* \in \cI_{\cF,k}$, it holds that $Q^* \cap X \in \cI(I \cap X)$ and  $Q^* \setminus X \in \cI(I \setminus X)$ . Thus,  
	\begin{equation*}
		\label{eq:k2}
		\begin{aligned}
			\textnormal{\textsf{DP}}_I(q,t) =  c(Q^*) 
			= c(Q^* \cap X)+c(Q^* \setminus X) 
			\geq \textnormal{\textsf{DP}}_{I \cap X}(q^*_1,t^*_1) + \textnormal{\textsf{DP}}_{I \setminus X}(q^*_2,t^*_2) 
			\geq  M. 
		\end{aligned}
	\end{equation*} 
	The first inequality holds since $\textnormal{\textsf{DP}}_{I \cap X}(q^*_1,t^*_1)$ is the minimum cost of $Q \in \cI(I \cap X)$ such that $|Q| = q^*_1$ and $p(Q) = t^*_1$; since $Q^* \cap X$ satisfies these conditions (by Observation~\ref{obs:cap2}), we conclude that $c(Q^* \cap X) \geq \textnormal{\textsf{DP}}_{I \cap X}(q^*_1,t^*_1) $. Similar arguments show that $c(Q^* \setminus X) \geq \textnormal{\textsf{DP}}_{I \setminus X}(q^*_2,t^*_2)$. The second inequality holds since $M$ is the minimum value of $\textnormal{\textsf{DP}}_{I \cap X}(q_1,t_1) + \textnormal{\textsf{DP}}_{I \setminus X}(q_2,t_2)$ over all $q_1,q_2,t_1,t_2 \in \mathbb{N}$ such that $q_1+q_2 = q, t_1+t_2 = t$. As $q^*_1,q^*_2,t^*_1,t^*_2 \in \mathbb{N}$ and $q^*_1+q^*_2 = q, t^*_1+t^*_2 = t$, the inequality follows. 
	\end{claimproof}

	\begin{claim}
		\label{claim:2}
		If $M \neq \infty$ and $q \leq k(S)$ then $\textnormal{\textsf{DP}}_I(q,t) \leq M$. 
	\end{claim}
	\begin{claimproof}
		Let $q'_1,q'_2,t'_1,t'_2 \in \mathbb{N}, q'_1+q'_2 = q, t'_1+t'_2 = t$ such that  \begin{equation*}
			\label{eq:arg2}
			\textnormal{\textsf{DP}}_{I \cap X}(q'_1,t'_1) + \textnormal{\textsf{DP}}_{I \setminus X}(q'_2,t'_2) = M.
		\end{equation*} 
		Since $M \neq \infty$, 
		there exist
		\begin{equation}
			\label{eq:arg3}
			Q'_1 =  \argmin_{Q \in \cI(I \cap X) \textnormal{ s.t. } |Q| = q'_1, ~p(Q) = t'_1} c(Q),~~~~~~~~~~~~~~~~~	Q'_2 =  \argmin_{Q \in \cI(I \setminus X) \textnormal{ s.t. } |Q| = q'_2, ~p(Q) = t'_2} c(Q).
		\end{equation} 
		By \eqref{eq:arg3}, 
		\begin{equation*}
			\label{eq:k1}
			\begin{aligned}
				M ={} & \textnormal{\textsf{DP}}_{I \cap X}(q'_1,t'_1) + \textnormal{\textsf{DP}}_{I \setminus X}(q'_2,t'_2) \\
				={} & c(Q'_1)+c(Q'_2) \\
				={} & c(Q'_1 \uplus Q'_2) \\
				\geq{} & \min_{Q \in \cI_{\cF,k} \textnormal{ s.t. } |Q| = q, ~p(Q) = t} c(Q) \\
				={} & \textsf{DP}_I(q,t). 
			\end{aligned}
		\end{equation*}  
	The above inequality follows from the next arguments. As $q'_1+q'_2 = q$ and $q \leq k(S)$, by Observation~\ref{obs:cap2} $Q'_1 \uplus Q'_2 \in \cI_{\cF,k}$; also,  $|Q'_1 \uplus Q'_2| = q$, and $p(Q'_1 \uplus Q'_2) = t'_1+t'_2 = t$. The inequality follows, as the minimum in the right-hand side of the inequality considers $Q = Q'_1 \uplus Q'_2$ in particular. 
	\end{claimproof}

	To complete the proof of the lemma, we consider four complementary cases. \begin{itemize}
		\item 
		If $q> k(S)$ then $\textsf{DP}_I(q,t) = \infty$ by Definition~\ref{def:DP}.
		\item 
		If $M = \infty$ then $\textsf{DP}_I(q,t) = \infty$ by Claim~\ref{claim:1}.
		\item
		 If $ q  \leq k(S)$ and $\textsf{DP}_I(q,t) = \infty$ then by Claim~\ref{claim:2} $M = \infty$ . 
		\item 
		If $ q  \leq k(S)$ and $\textsf{DP}_I(q,t) \neq \infty$ then, by Claims~\ref{claim:1} and~\ref{claim:2}, $\textsf{DP}_I(q,t) = M$. 
	\end{itemize}
\end{proof}

Note that computing the table $\textsf{DP}_I$ by Lemma~\ref{lem:DP-step} is possible only if there is a maximal set $X \in \cF \setminus \{S\}$; this requires more than one set in the laminar family. If $|\cF| = 1$, we compute the table $\textsf{DP}_I$ using two alternative ways, depending on whether $|S| = 1$ or $|S|>1$. 
The next observation considers the case where the instance consists of a single element; it follows immediately from Definition~\ref{def:DP}. 

\begin{obs}
	\label{obs:baseCase}
	For a \textnormal{BLM} instance $I = (S,\cF,k,c,p,B)$ such that $|S| = 1$ it holds that $\textnormal{\textsf{DP}}_I(1,p(S)) = c(S)$,  $\textnormal{\textsf{DP}}_I(0,0) =  0$, and for any other $(q,t) \in \{0,\ldots, |S|\} \times P_I$ it holds that  $\textnormal{\textsf{DP}}_I(q,t) =  \infty$.  
\end{obs}

We now consider the case where 
$|S|>1$ 
and $\cF = \{S\}$. Here, we define a new instance which adds a (redundant) partition of S into two subsets. This partition allows us to use the recursive computation as
given in Lemma~\ref{lem:DP-step}. 
For a \textnormal{BLM} instance $I = (S,\cF,k,c,p,B)$ such that $\cF = \{S\}$ and $|S|>1$, we say that the BLM instance $\tilde{I} = (S,\tilde{\cF},\tilde{k},c,p,B)$ is a {\em partitioned-instance} of $I$ if the following holds. \begin{itemize}
	\item $\tilde{\cF} = \{S_1,S_2,S\}$ where $S_1,S_2$ is a partition of $S$. 
	
	\item $\tilde{k}: \tilde{\cF} \rightarrow \mathbb{N}$ and $\tilde{k}(S_1) = \tilde{k}(S_2) = \tilde{k}(S) = k(S)$. 
\end{itemize}

Note that $S_1 ,S_2 \neq \emptyset$ since $|S|>1$. The next lemma states that the independent sets of an instance $I$ and 
of a partitioned-instance of $I$ are identical. 
\begin{lemma}
	\label{lem:auxP}
	For any \textnormal{BLM} instance $I =  (S,\cF,k,c,p,B)$ such that $\cF = \{S\}$ and $|S|>1$, and a partitioned-instance $\tilde{I} = (S,\tilde{\cF},\tilde{k},c,p,B)$ of $I$ it holds that $\cI_{\cF,k} = \cI_{\tilde{\cF},\tilde{k}}$. 
\end{lemma}

\begin{proof}
	Let $\tilde{\cF} = \{S_1,S_2,S\}$ and
	 $Q \in \cI_{\cF,k}$. Then $|Q| \leq k(S) = \tilde{k}(S)$, and
	 $$|Q \cap S_1|  \leq |Q| \leq k(S) = \tilde{k}(S_1).$$
	 Similarly, $ |Q \cap S_2| \leq \tilde{k}(S_2)$. Thus, $Q \in \cI_{\tilde{\cF},\tilde{k}}$. For the other direction, let 
	$ \tilde{Q} \in \cI_{\tilde{\cF},\tilde{k}}$. 
	Then, $|\tilde{Q}| \leq \tilde{k}(S) = k(S)$ and it follows that $\tilde{Q} \in \cI_{\cF,k}$. We conclude that  $\cI_{\cF,k} = \cI_{\tilde{\cF},\tilde{k}}$.
\end{proof}

The next result follows immediately from Lemma~\ref{lem:auxP} and  Definition~\ref{def:DP}. 

	\begin{cor}
	\label{lem:partitioned}
	For any \textnormal{BLM} instance $I$ such that $\cF = \{S\}$ and $|S|>1$, and a partitioned-instance $\tilde{I}$ of $I$, it holds that $\textnormal{\textsf{DP}}_{\tilde{I}} =\textnormal{\textsf{DP}}_I$. 
\end{cor}

	\begin{figure}[htbp]
	\hspace*{0.25cm}                                                           
	\includegraphics[scale=0.33]{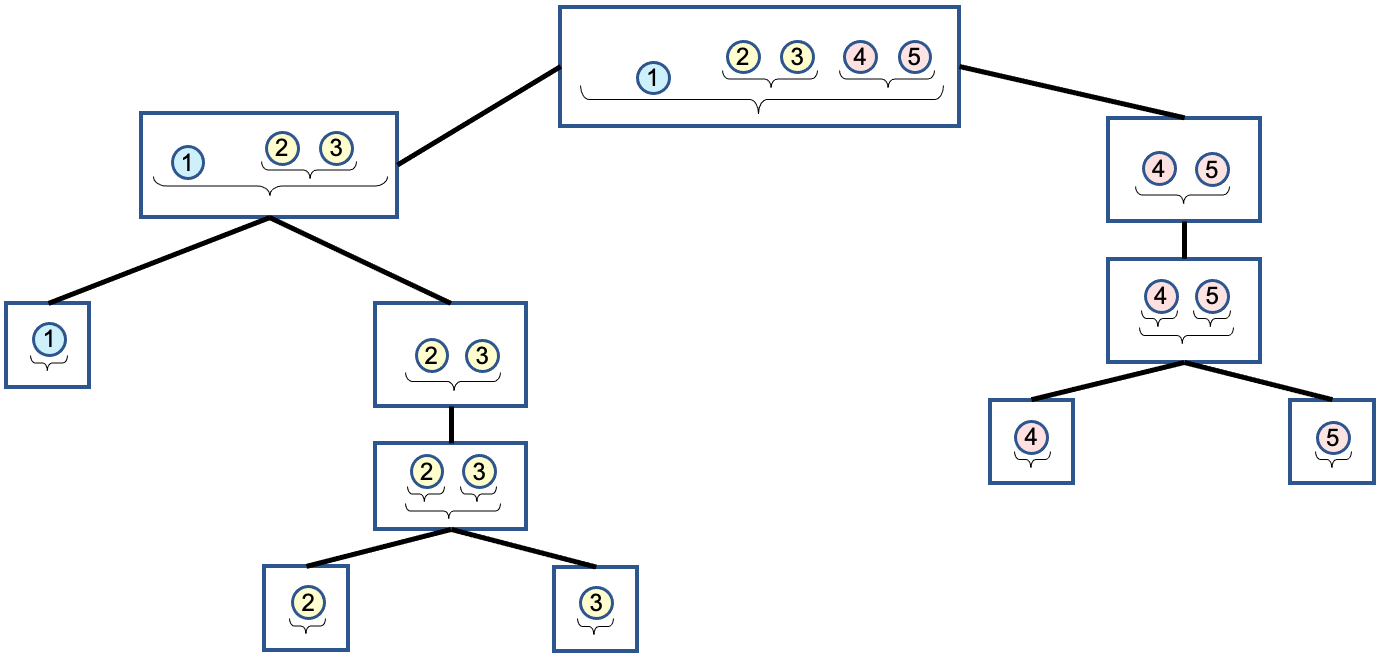}%
	\hspace{1mm}%
	\caption{An illustration of the recursive calls of Algorithm~\ref{alg:DP} for a BLM instance $I$ with five elements $S = \{1,2,3,4,5\}$, represented by the circles. 
The underbraces indicate the sets in the laminar family $\cF$ of the instance in each node of the tree. For example, in the root node, 	
$\cF = \{S,\{2,3\},\{4,5\}\}$.}
\label{fig:1}
\end{figure}

Using the above, we derive a pseudo-polynomial time algorithm which computes $\textsf{DP}_I$. If there is one element, the algorithm computes $\textsf{DP}_I$ using Observation~\ref{obs:baseCase}; otherwise, if there is one set $\cF = \{S\}$ in the laminar family, the algorithm computes $\textsf{DP}_I$ by a recursive call to the algorithm with a partitioned-instance. The remaining case is that there exists a maximal set $X \in \cF \setminus \{S\}$ for which we can apply the recursive computation of the table via Lemma~\ref{lem:DP-step}. We give an illustration of recursive calls the algorithm initiates in Figure~\ref{fig:1}. The pseudocode of the algorithm is given in Algorithm~\ref{alg:DP}.  For a BLM instance $I$, we use $d(I)$ to denote the recursion depth in the execution of $\textsf{ComputeDP}\left(I\right)$. 

\begin{algorithm}[h]
	\caption{$\textsf{ComputeDP}(I)$}
	\label{alg:DP}
	\SetKwInOut{Input}{input}
	\SetKwInOut{Output}{output}
	
	\Input{ A BLM instance $I = (S,\cF,k,c,p,B)$.}

\Output{The table $\textsf{DP}_I$ as defined in Definition~\ref{def:DP}.}

Initialize $\textnormal{\textsf{DP}}_I(q,t) =\infty$
for all $q \in  \{0,1,\ldots, |S|\}$ and $t \in P_I$.\label{set:initialize}

\If{$|S| = 1$}{

Set $\textnormal{\textsf{DP}}_I(1,p(S)) \leftarrow c(S)$ and  $\textnormal{\textsf{DP}}_I(0,0) \leftarrow 0$.\label{step:set0}

Return the table $\textnormal{\textsf{DP}}_I$. 

}

\If{$|\cF| = 1$}{
	
	 
	Return $\textsf{ComputeDP}\left(\tilde{I}\right)$ where $\tilde{I}$ is a partitioned-instance of $I$.\label{step:partitioned}
}


Find a maximal set $X \in \cF \setminus \{S\}$ 
.\label{step:find}

Compute $\textsf{DP}_{I \cap X} \leftarrow \textsf{ComputeDP}(I \cap X)$ and $\textsf{DP}_{I \setminus X} \leftarrow \textsf{ComputeDP}(I \setminus X)$.\label{step:comp}

For $q \in  \{0,1,\ldots, |S|\}$ and $t \in P_I$ set 
\begin{equation*}
	\textnormal{\textsf{DP}}_I(q,t) =
	\begin{cases}
	 \min_{\substack{ q_1,q_2,t_1,t_2 \in \mathbb{N} \textnormal{ s.t. } \\ q_1+q_2 = q, t_1+t_2 = t }}  \left\{ \textnormal{\textsf{DP}}_{I \cap X}(q_1,t_1) + \textnormal{\textsf{DP}}_{I \setminus X}(q_2,t_2) \right\} ~~ & \textnormal{if}\ q \leq k(S) \\
		\infty ~~& \textnormal{otherwise}
	\end{cases}
\end{equation*}\label{step:cc-step}

Return the table $\textnormal{\textsf{DP}}_I$. 
\end{algorithm}

\begin{lemma}
	\label{lem:main-DP}
	For any \textnormal{BLM} instance $I = (S,\cF,k,c,p,B)$ Algorithm~\ref{alg:DP} returns 
	the table $\textnormal{\textsf{DP}}_{I}$. 
\end{lemma}

\begin{proof}
	We show that for every BLM instance $I = (S,\cF,k,c,p,B)$ it holds that $\textsf{ComputeDP}\left(I\right) = \textsf{DP}_I$. The proof is by induction on $d(I)$. 
	For the base case, let $I$ be a BLM instance such that $d(I) = 1$. Then, it holds that $|S| = 1$ and the algorithm returns $\textsf{DP}_I$ by Observation~\ref{obs:baseCase} and Steps~\ref{set:initialize}, \ref{step:set0}
	of Algorithm~\ref{alg:DP}. For some $n \in \mathbb{N}_{>0}$, assume that for every BLM instance $I'$ for which $d(I') \leq n$, it holds that $\textsf{ComputeDP}\left(I'\right) = \textsf{DP}_{I'}$. For the induction step, let $I = (S,\cF,k,c,p,B)$ be a BLM instance such that $d(I) = n+1$. We consider two cases.
	\begin{enumerate}
	
		\item  $|\cF| = 1$. Then Algorithm~\ref{alg:DP} returns $\textsf{ComputeDP}\left(\tilde{I}\right)$ by Step~\ref{step:partitioned}, where $\tilde{I}$ is a partitioned-instance of $I$. By the induction hypothesis, it holds that $\textsf{ComputeDP}\left(\tilde{I}\right) = \textsf{DP}_{\tilde{I}}$. Thus, by Corollary~\ref{lem:partitioned}, we have $\textsf{ComputeDP}\left(I\right)  = \textsf{ComputeDP}\left(\tilde{I}\right) = \textsf{DP}_{I}$.

		\item $|\cF|>1$. 
		In Step~\ref{step:find} in the (non-recursive) computation of  $\textsf{ComputeDP}\left(I\right)$, the algorithm finds a maximal set $X \in \cF \setminus \{S\}$. 
		By the assumption of the induction it follows that  $\textsf{ComputeDP}(I \cap X) = \textsf{DP}_{I \cap X}$ and $ \textsf{ComputeDP}(I \setminus X) = \textsf{DP}_{I \setminus X}$. Therefore, by Step~\ref{step:cc-step} and Lemma~\ref{lem:DP-step} it holds that $\textsf{ComputeDP}\left(I\right) = \textsf{DP}_{I}$. 
	\end{enumerate}

\end{proof}

For the running time of Algorithm~\ref{alg:DP}, assume that the laminar family is represented by a linked list of sets, and that the elements in each set are represented by a bit map.

\begin{lemma}
	\label{lem:main-DPTIME}
	For any \textnormal{BLM} instance $I = (S,\cF,k,c,p,B)$, the running time of Algorithm~\ref{alg:DP} on $I$ is $O(|S|^3 \cdot |P_I|^2)$.
\end{lemma}

\begin{proof}
		 We use the next claim. 
	\begin{claim}
	\label{claim:auxh2}
	For any \textnormal{BLM} instance $I = (S,\cF,k,c,p,B)$, Algorithm~\ref{alg:DP} makes at most $3 \cdot |S|$ recursive calls during the execution of $\textnormal{\textsf{ComputeDP}}(I)$.
\end{claim}
\begin{claimproof}
	Consider the tree of recursive calls generated throughout the execution of 
	$\textnormal{\textsf{ComputeDP}}(I)$. Each element in $S$ has a unique leaf in the tree; therefore, the number of leaves is bounded by $|S|$. Moreover, since there are $|S|$ leaves, the number of internal nodes in the tree that have two children is bounded by $|S|-1$; thus, the number of recursive calls initiated in Step~\ref{step:comp} is at most $|S|$. Finally, 
	after each recursive call from Step~\ref{step:partitioned} the algorithm applies a recursive call from Step~\ref{step:comp}.  Therefore, the number of recursive calls from Step~\ref{step:partitioned} is at most $|S|$. Overall, Algorithm~\ref{alg:DP} makes at most $3 \cdot |S|$ recursive calls.
\end{claimproof}

To complete the proof, we show that the running time of the algorithm, excluding the recursive calls, is $O(|S|^2 \cdot |P_I|^2)$. As the recursive calls in the algorithm use instances in which the number of elements is bounded by $|S|$, and the set of profits is of size at most $P_I$, the statement of the lemma follows from Claim~\ref{claim:auxh2}.
 Computing Step~\ref{set:initialize} takes $O(|S| \cdot |P_I|)$. Moreover, Step~\ref{step:partitioned} can be computed in time $O(|S|)$ using an arbitrary partition of the elements. Step~\ref{step:find} can be computed in time $O(|S|)$ by iterating over all sets in the laminar family $\cF$. Also, computing each entry in the table $\textsf{DP}_I$ in Step~\ref{step:cc-step} takes $O(|S| \cdot |P_I|)$; thus, computing the entire table takes $O(|S|^2 \cdot |P_I|^2)$. Overall, the running time is $O(|S|^3 \cdot |P_I|^2)$. 
\end{proof}

\section{An FPTAS for BLM}
\label{sec:FPTAS}

In this section we use the dynamic program in Section~\ref{sec:alg} to derive an FPTAS for BLM, leading to the proof of Theorem~\ref{thm:main}. Let $I = (S,\cF,k,c,p,B)$ be a BLM instance and let $\eps>0$ be an error parameter. 
Note that the computation time of $\textsf{DP}_I$ depends on $|P_I|$, may not be polynomial in the input size. To obtain a polynomial-time algorithm, we round down the profit of each item $e$ to $\floor{\frac{p(e)}{\alpha}}$, where $\alpha = \frac{\eps \cdot \max_{e \in S} p(e)}{|S|}$. This generates a reduced instance $\bar{I}$,
for which the table $\textsf{DP}_I$ can be computed efficiently. 
Then, by iterating over all possible values in $\textsf{DP}_I$, we compute the value of the optimum for $\bar{I}$; this gives an {\em almost} optimal solution for $I$, where the solution itself is computed using standard backtracking. The pseudocode of the algorithm is given in Algorithm~\ref{alg:FPTAS}. 

\begin{algorithm}[h]
	\caption{$\textsf{FPTAS}(I,\eps)$}
	\label{alg:FPTAS}
	\SetKwInOut{Input}{input}
	\SetKwInOut{Output}{output}
	
	\Input{A BLM instance $I = (S,\cF,k,c,p,B)$ and an error parameter $\eps>0$.}
	
	\Output{A solution $T$ of $I$ with profit $p(T) \geq (1-\eps) \cdot \OPT(I)$.}

	
	Let $\alpha \leftarrow \frac{\eps \cdot \max_{e \in S} p(e)}{|S|}$ and define $\bar{p}(e) \leftarrow \floor{\frac{p(e)}{\alpha}} ~\forall e \in S$.\label{step:round}

	Compute $\textsf{DP}_{\bar{I}} \leftarrow \textsf{ComputeDP}(\bar{I})$, where $\bar{I} = (S,\cF,k,c,\bar{p},B)$.\label{step:table}
	
	Let $\lambda \leftarrow  \bigg\{(q,t) \in \{0,\ldots, |S|\} \times P_{\bar{I}}~\bigg|~ \textsf{DP}_{\bar{I}}(q,t) \leq B\bigg\}$.\label{step:lam1}
	
	Use backtracking to find a solution $T$ for $I$ of value $\max_{(q,t) \in \lambda} t$.\label{step:lam2}
\end{algorithm}

\noindent{\bf Proof of Theorem~\ref{thm:main}:}  We show that Algorithm~\ref{alg:FPTAS} is an FPTAS for BLM. Let $\bar{I} = (S,\cF,k,c,\bar{p},B)$ be the instance with the rounded profits as given in Algorithm~\ref{alg:FPTAS}. 
By Lemma~\ref{lem:main-DP}, $\textsf{ComputeDP}(\bar{I})$ returns the table $\textsf{DP}_{\bar{I}}$ as given in Definition~\ref{def:DP}. 
 Thus, for all $(q,t) \in \lambda$ there is a solution for $\bar{I}$ of profit $t$ if and only if $\textsf{DP}_{\bar{I}}(q,t) \leq B$. 
 Let $T$ be an optimal solution for $\bar{I}$. By the above, we have  
\begin{equation}
	\label{eq:mu}
	\bar{p}(T) = \max_{(q,t) \in \lambda} t = \OPT(\bar{I}) 
\end{equation} 

Let $T^*$ be an optimal solution of $I$. Then,

\begin{equation}
	\label{eq:opt}
	\begin{aligned}
		\bar{p}(T^*) = \sum_{e \in T^*} \floor{\frac{p(e)}{\alpha}} \geq  \sum_{e \in T^*} \left(\frac{p(e)}{\alpha}-1\right) \geq  \sum_{e \in T^*} \left(\frac{p(e)}{\alpha}\right) - |S| = \frac{\OPT(I)}{\alpha}-|S|. 
	\end{aligned}
\end{equation} 
Hence, 
\begin{equation*}
	\label{eq:f}
	p(T) \geq \alpha \cdot \bar{p}(T) \geq \alpha \cdot \bar{p}(T^*) \geq \alpha \cdot \left(  \frac{\OPT(I)}{\alpha}-|S|  \right) = \OPT(I)-\eps \cdot \max_{e \in S} p(e) \geq (1-\eps) \cdot \OPT(I). 
\end{equation*} 
The first inequality holds by Step~\ref{step:round}. The second inequality follows from the optimality of $\bar{p}(T)$ by \eqref{eq:mu}. The third inequality holds by \eqref{eq:opt}. The last inequality holds since $\OPT(I) \geq \max_{e \in S} p(e)$ (assuming that $c(e) \leq B ~\forall e \in S$). 

We now analyze  the running time of the scheme. 
We note that Step~\ref{step:round} takes linear time. 
The running time of Step~\ref{step:table} is $O(|S|^3 \cdot |P_{\bar{I}}|^2)$ by Lemma~\ref{lem:main-DPTIME}. Finally, Steps~\ref{step:lam1} and~\ref{step:lam2} can be computed in time $O(|S| \cdot|P_{\bar{I}}|)$ by the definition of $\lambda$. Hence, the overall running time is $$O\left(|S|^3 \cdot |P_{\bar{I}}|^2\right) = O\left(|I^5| \cdot \eps^{-2}\right).$$
The equality follows since $|\bar{P}_{\bar{I}}| = O(|S| \cdot \eps^{-1})$ by Step~\ref{step:round}. \hfill \qedsymbol

\section{Discussion}
\label{sec:discussion}

In this paper we showed that the budgeted laminar matroid independent set (BLM) problem admits an FPTAS, thus improving upon the existing EPTAS for this matroid family, and generalizing the FPTAS for the special cases of cardinality constrained knapsack and multiple-choice knapsack. Our FPTAS is based on a natural dynamic program which utilizes the tree-like structure of laminar matroids. It seems that 
with slight modifications our scheme yields an FPTAS for the more general problem of 
budgeted $k$-laminar matroid independent set, where $k \in \mathbb{N}$ is fixed.\footnote{For a definition of $k$-laminar matroids see, e.g.~\cite{fife2019generalized}.}

An intriguing open question is whether BMI admits an FPTAS on other families of matroids, such as graphic matroids, transversal matroids, or linear matroids. 

We note that the running time of our scheme is $O(|I|^5 \cdot \eps^{-2})$, whereas the running time of the state of the art FPTAS for knapsack is $O(|I| + \eps^{-2.2})$ \cite{deng2023approximating}, which almost matches the lower bound of $\Tilde{O}\left((|I| + \eps^{-1})^{2-o(1)}\right)$ for the problem \cite{cygan2019problems}. 
It would be interesting to design an FPTAS for BLM that matches the running time of \cite{deng2023approximating}, or to obtain a stronger lower bound for this problem.

\comment{Compare the running time of the algorithm to the state of art result and lower bound for knapsack. Mention that it is not clear if a stronger lower bound can be obtained for BLM, or if one can obtain a running time which matches that of the knapsack problem.
	
	Mention that the lower bound (no fptas) for budgeted matroid IS does not rule out the existence of FPTAS for other classes of matroids. For example, graph matroids and linear matroids. Mention the existence of an fptas for those special cases of budgeted matroid independent set is open. Are there classes of matroids which are generalizations of laminar matroids?
	
	I think we can define a generalization of laminar matroids for which our algorithm may work. The idea is to think of laminar families as families of set defined by graphs of tree-width 1. What if we increase the tree-width to some bounded tree-width? It may be worth formalizing the definitions and writing things down.}
\bibliographystyle{splncs04}
\bibliography{bibfile}
\end{document}